\pdfoutput=1

\documentclass{article}

\usepackage[letterpaper,margin=1in]{geometry}

\usepackage{balance} 
\usepackage{amsmath,amsthm}
\usepackage{amssymb}
\usepackage{booktabs}
\usepackage{tabularx}
\usepackage{array}
\usepackage{enumitem}
\usepackage{graphicx}
\usepackage{color}
\usepackage{xspace}
\usepackage[numbers]{natbib}
\usepackage[hidelinks]{hyperref}
\usepackage{verbatim}

\usepackage[linesnumbered,ruled,vlined]{algorithm2e}
\usepackage{multirow}

\newcommand{\Description}[1]{} 

\title{Bribery's Influence on Ranked Aggregation}
\author{%
  Pallavi Jain\\
  Indian Institute of Technology Jodhpur\\
  \texttt{pallavi@iitj.ac.in}
  \and
  Anshul Thakur\\
  Indian Institute of Technology Jodhpur\\
  \texttt{thakur.17@iitj.ac.com}%
}
\date{} 

\newcommand{\BibTeX}{\rm B\kern-.05em{\sc i\kern-.025em b}\kern-.08em\TeX}

\newtheorem{theorem}{Theorem}
\newtheorem{lemma}[theorem]{Lemma}
\newtheorem{corollary}[theorem]{Corollary}

\newtheorem{obs}{Observation}

\newtheorem{clm}[theorem]{Claim}
\newtheorem{example}{Example}
\newcommand{\defproblem}[3]{%
  \begin{center}
  \fbox{\begin{minipage}{0.97\linewidth}
    \textbf{#1}\par\smallskip
    \textbf{Input:} #2\par
    \textbf{Question:} #3
  \end{minipage}}%
  \end{center}
}

\newcommand{\no}{{\sc No}\xspace}
\newcommand{\yes}{{\sc Yes}\xspace}
\newcommand{\opt}{{\mathtt{OPT}\xspace}\xspace}

\newcommand{\nph}{\textsf{\textup{NP-hard}}\xspace}

\newcommand{\dollarbribery}{\textsc{$\$$-{\sc Bribery}-{\sc Kemeny Score}}\xspace}
\newcommand{\swapbribery}{{\sc Swap Bribery-{Kemeny Score}}\xspace}

\newcommand{\possiblewinners}{\textsc{Possible  Kemeny Score}\xspace}
\newcommand{\voterdeletion}{{\sc Ranking Deletion-Kemeny Score}\xspace}
\newcommand{\candidatedeletion}{{\sc Candidate Deletion-Kemeny Score}\xspace}
\newcommand{\kemenyscore}{\textsc{Kemeny Score}\xspace}
\newcommand{\kemenyconsensus}{\textsc{Kemeny Consensus}\xspace}
\newcommand{\kemenywinner}{\textsc{Kemeny Winner}\xspace}
\newcommand{\mwcis}{\textsc{HCIS}\xspace}

\newcommand{\M}{\mathfrak{B}\xspace}
\newcommand{\pname}
{\textsc{$\M$-Kemeny Score}\xspace}

\newcommand{\possibleks}
{\textsc{Possible Kemeny Score}\xspace}

\DeclareMathOperator{\cost}{\mathtt{cost}\xspace}

\newcommand{\budget}{\mathtt{bgt}\xspace}
\newcommand{\kt}{{\mathtt{KT}}\xspace}
\newcommand{\R}{\Co{R}\xspace}
\newcommand{\X}{X \xspace}
\newcommand{\C}{\Co{C}\xspace}
\newcommand{\instanceI}{\mathcal{I}}
\newcommand{\instanceJ}{\mathcal{J}}

\newcommand{\Co}[1]{\ensuremath{\mathcal{#1}}\xspace}

\newcommand{\hide}[1]{}
\newcommand{\dist}[2]{\kt(#1, #2)}

\newcommand{\disagree}[3]{\ensuremath{\#\mathsf{disagree}(#1,#2,#3)}}
\newcommand{\insertop}[3]{\ensuremath{\mathsf{insert}(#1, #2, #3)}}

\newcommand{\tmpR}[0]{A\xspace}

\DeclareMathOperator*{\argmin}{arg\!min}

\begin{document}



\maketitle

\begin{abstract}
Kemeny Consensus is a well-known rank aggregation method in social choice theory. In this method, given a set of rankings, the goal is to find a ranking $\Pi$ that minimizes the total Kendall tau distance to the input rankings. Computing a Kemeny consensus is {\sf NP}-hard, and even verifying whether a given ranking is a Kemeny consensus is {\sf coNP}-complete. Fitzsimmons and Hemaspaandra [IJCAI 2021] established the computational intractability of achieving a desired consensus through manipulative actions. Kemeny Consensus is an optimisation problem related to Kemeny's rule. In this paper, we consider a decision problem related to Kemeny's rule, known as \emph{Kemeny Score}, in which the goal is to decide whether there exists a ranking $\Pi$ whose total Kendall tau distance from the given rankings is at most $k$. Computation of Kemeny score is known to be {\sf NP}-complete. In this paper, we investigate the impact of several manipulation actions on the Kemeny Score problem, in which given a set of rankings, an integer $k$, and a ranking $X$, the question is to decide whether it is possible to manipulate the given rankings so that the total Kendall tau distance of \X from the manipulated rankings is at most $k$. We show that this problem can be solved in polynomial time for various manipulation actions. Interestingly, these same manipulation actions are known to be computationally hard for Kemeny consensus.
\end{abstract}



\section{Introduction}
Rank aggregation, the task of combining multiple rankings into a single consensus ranking, is a ubiquitous problem with wide-ranging applications. It is a fundamental problem in computational social choice theory and plays a crucial role in various domains such as meta-search engines~\cite{DBLP:conf/www/DworkKNS01}, recommendation systems, machine learning~\cite{balchanowski2023comparative}, marketing and advertisement, bioinformatics~\cite{li2019comparative}, etc. One of the rank aggregation techniques, widely accepted in social choice theory, is \emph{Kemeny's} rule defined by Kemeny in 1959~\cite{kemeny1959mathematics}. In the \kemenyscore problem, given a set of $n$ rankings, say $\R$, and an integer $k$, the goal is to determine if there exists a ranking $\X$ such that the total \emph{Kendall tau distance} (defined later) between the set of rankings in $\R$ and $\X$ is at most $k$~\cite{brandt2016handbook,DBLP:journals/tcs/HemaspaandraSV05}.   Kemeny's rule is interesting in social choice theory as it satisfies Condorcet extension criterion~\cite{truchon1998extension} and is also maximum likelihood estimator for a simple probabilistic model in which individual preferences are noisy estimates of an underlying ``true'' ranking~\cite[page 87]{brandt2016handbook}. \kemenyscore is well-studied computationally from theoretical computer science viewpoint~\cite{DBLP:journals/tcs/HemaspaandraSV05,DBLP:conf/caldam/DeMDM24,DBLP:journals/aamas/BetzlerBN14,DBLP:journals/tcs/BetzlerFGNR09,DBLP:conf/atal/BetzlerFGNR09,DBLP:journals/dm/BiedlBD09,DBLP:conf/stoc/Kenyon-MathieuS07,DBLP:conf/aaai/FominLRS10}, machine learning~\cite{DBLP:conf/aaai/GeorgeD24}, as well as heuristically~\cite{DBLP:journals/eor/RicoVD23,DBLP:journals/jaihc/RicoVD24,DBLP:conf/aaai/DavenportK04}. Interestingly, \kemenyscore is polynomial-time solvable for two rankings and \nph for four rankings~\cite{DBLP:journals/dm/BiedlBD09}, while the case of three rankings is a long-standing open problem.

One of the important directions of research in computational social choice theory is to study the computational complexity of changing the outcome via manipulation, bribery, and various control actions, such as adding or deleting candidates or rankings~\cite{brandt2016handbook,DBLP:journals/tcs/AnandD21,DBLP:journals/tcs/DeyMN18,DBLP:conf/atal/KusekBF0K23,DBLP:journals/teco/KnopKM20,DBLP:journals/aamas/FaliszewskiRRS15,DBLP:journals/jair/FaliszewskiHH09}. To the best of our knowledge, de Haan~\cite{DBLP:conf/atal/Haan17} initiated the study of manipulation for Kemney in the context of judgement aggregation. Recently, Fitzsimmons and Hemaspaandra~\cite{DBLP:conf/ijcai/FitzsimmonsH21} studied the computational complexity of manipulation schemes for the \kemenyconsensus problem, where given the set of $n$ rankings, $\R$, and a ranking $\X$, the goal is to manipulate the rankings $\R$ using a manipulative action so that $\X$ attains the minimum total Kendall tau distance to the profile among all rankings. The Kendall tau distance between two rankings $\pi$ and $\pi'$ over a candidate set $\C$ is the number of pairs of candidates whose relative ordering differs between $\pi$ and $\pi'$. Mathematically, 

\[\kt(\pi,\pi')=|\{\{x,y\}\subseteq \C \colon (x \prec_\pi y  \wedge  y \prec_{\pi'} x) \vee (y \prec_\pi x \wedge x \prec_{\pi'} y)\}|\]

The notation $x \prec_\pi y$ denotes that $x$ is before $y$ in the ranking $\pi$.
The Kendall tau distance of a ranking $\pi$ from the set of rankings $\R$ is the summation of distance from $\pi$ to each ranking in $\R$, i.e., $\sum_{\pi' \in \R}\kt(\pi,\pi')$. 

In this paper, we explore the computational complexity of various manipulation schemes for \kemenyscore, which has not been studied so far to the best of our knowledge.

\paragraph{\bf \kemenyconsensus vs \kemenyscore}\kemenyscore is a decision problem of Kemeny's rule and  \kemenyconsensus is an optimisation problem of Kemeny's rule~\cite[page 88]{brandt2016handbook}.  However, interestingly, when manipulation is considered, there is a significant difference between the two problems. In the manipulation problems for \textsc{Kemeny Score}\xspace, the Kendall tau distance of $\X$ should be at most $k$ from $\Co{R}$. At the same time, in manipulation problems for \textsc{Kemeny Consensus}\xspace it is possible that the distance of $\X$ from $\Co{R}$ does not change at all, rather the distance of other rankings increases so that $\X$ is the closest one. We explain it using the following example, which is the same as Example 9 in~\cite{DBLP:conf/ijcai/FitzsimmonsH21}.
\begin{example}
    Consider three candidates $a,b,c$. Suppose there are three rankings $R_1=R_2=R_3= a \succ b \succ c$, two rankings $R_4=R_5=a\succ c \succ b$, and one ranking $R_6 = c \succ b \succ a$. Let $X= a\succ c \succ b$. Let $k=3$. Note that if we delete ranking $R_6$, then the Kendall tau distance of $\X$ from $\R=\{R_1,\ldots,R_5\}$ is $3$; while $\X$ is not a Kemeny consensus. The Kemeny consensus ranking is $a \succ b \succ c$. However, if we delete $R_1$, then $\X$ is the Kemeny consensus of $\R=\{R_2,\ldots,R_6\}$ while the Kendall tau distance of $\X$ from $\R$ is $4$.   
\end{example}
    
 While many problems, including the \kemenyconsensus and the {\sc Kemeny Winner} (the top candidate in the Kemeny consensus is the winner), are not susceptible to various bribery-based manipulation schemes, we have developed polynomial-time algorithms for the {\sc Kemeny Score} under various bribery-based manipulative actions. 
 Formally, we study the following problem.



\defproblem{\pname}{a set of candidates $\C$, a set of rankings ${\R}$ over $\C$, a ranking $\X$ over $\C$, a cost function $\cost_{\M}$ (depending on the bribery action $\M$), a budget $\budget \in \mathbb{N}$, and an integer $k$.}{Does there exist a set of rankings ${\R}'$ that can be obtained by manipulating rankings in ${\R}$ using bribery action $\M$ within the budget $\budget$, and the total Kendall tau distance of $\X$ from the rankings in ${\R}'$ is at most $k$?}

We consider the following bribery actions. The cost of bribery is additive for all bribery actions.

\begin{description}
    \item[$\$$Bribery.] In this manipulation action, a ranking $R\in \R$ can be changed to any arbitrary ranking. Note that for our problem, changing them to $\X$ is a better choice rather than anything else, as this makes the distance $0$. The cost function is defined as follows: $\cost_{\$}\colon \R \rightarrow \mathbb{N}$. Let $S\subseteq \R$ be the set of rankings that are changed to $\X$. Then, the cost of manipulation is $\sum_{R\in S}\cost_{\$}(R)$.    
    \item[Swap Bribery.] A swap $(c,c')_R$ in the ranking $R$ is called \emph{admissible} if $c$ and $c'$ are consecutive in $R$. In this manipulation action, we can change a ranking by admissible swaps, and the cost depends on the ranking and the number of swaps. Formally, the cost function is defined as follows: $\cost_{\sf SB}\colon \R  \rightarrow \mathbb{N}$. Note that $\cost_{\sf SB}(R)$ denotes the cost of one swap in the ranking $R$. If $l_1$ swaps are done on ranking $R_1$ and $l_2$ swaps are done on ranking $R_2$ then the total cost would be $l_1 \cdot \cost_{\sf SB}(R_1) + l_2 \cdot \cost_{\sf SB}(R_2)$
    \item[Ranking Deletion.] In this manipulation action, we are allowed to delete rankings, and there is a cost associated with the deletion. The cost function is defined as follows: $\cost_{\sf RD}\colon \R \rightarrow \mathbb{N}$.  
    \item[Candidate Deletion.] In this manipulation action, we are allowed to delete candidates. Note that if we delete a candidate $c \in \C$, then it is deleted from all the rankings in $\R$ as well as from $\X$, and it does not change the ordering of other candidates. The cost function is defined as follows: $\cost_{\sf CD}\colon \C \rightarrow \mathbb{N}$.
\end{description}

\emph{In the technical sections, we omit the subscript below $\cost$ since each section is dedicated to a specific bribery action. Thus, it will be clear from the context.}

Candidate deletion and ranking deletion are usually considered as control actions. Since we are considering the cost, we consider these manipulative actions also as bribery actions for brevity. 

Note that the addition of rankings or candidates is not a meaningful manipulation action for us. This is because adding rankings does not decrease the distance of $\X$ from the existing rankings. Additionally,  since $\X$ is given to us, the addition of candidates is not a meaningful operation for us.  

We also study the problem when the rankings in $\R$ are not complete, i.e., all the rankings in $\R$ are partial; however, $\X$ is a ranking over $\C$. The objective is to complete the rankings in $\R$ so that $\X$ is at most $k$ distance away from the new complete rankings over $\C$. Formally, we study the following problem. 

\defproblem{\possibleks}{a set of candidates $\C$, a set of partial rankings ${\R}$, a ranking $\X$ over $\C$, and an integer $k$.}{Can we \emph{extend}  rankings in $\R$ using the remaining candidates so that the total Kendall tau distance of $\X$ from the extended rankings in ${\R}$ is at most $k$?}

By \emph{extending} a ranking $R$, we mean that the candidates of $\C\setminus C'$ are added to $R$, where $R$ is over $C'$; however, the relative ordering of the candidates of the set $C'$ in the ranking $R$ does not change. 


Table~\ref{tab:summary} summarizes the complexity of various manipulation schemes for \kemenyscore, \kemenyconsensus, and {\sc Kemeny Winner}. All the results pertaining to \kemenyscore are build in this paper.

\begin{table}[!t]
  \centering
  \small
  \setlength{\tabcolsep}{4pt}
  \renewcommand{\arraystretch}{1.15}
  \begin{tabularx}{\linewidth}{|>{\raggedright\arraybackslash}p{0.22\linewidth}|>{\centering\arraybackslash}X|>{\centering\arraybackslash}X|>{\centering\arraybackslash}X|}
    \hline
    & \kemenyscore & \kemenyconsensus & {\sc Kemeny Winner} \\
    \hline
    {\bf $\$$Bribery} & P [Theorem~\ref{thm:dollar-bribery}] & {\sf coNP}-complete~\cite[Observation 6]{DBLP:conf/ijcai/FitzsimmonsH21} & OPEN \\
    \hline
    {\bf Swap Bribery} & P [Theorem~\ref{thm:swap-bribery}] & OPEN & OPEN \\
    \hline
    {\bf Ranking Deletion} & P [Theorem~\ref{thm:ranking-deletion}] & $\Sigma_2^p$-complete~\cite[Conjecture on page 199]{DBLP:conf/ijcai/FitzsimmonsH21} & OPEN \\
    \hline
    \multirow{3}{=}{\bf Candidate Deletion} & $k=0$: P [Theorem~\ref{thm:cand_del_polyk0}] & \multirow{3}{=}{$\Sigma_2^p$-complete~\cite[Theorem 8]{DBLP:conf/ijcai/FitzsimmonsH21}} & \multirow{3}{=}{$\Sigma_2^p$-complete~\cite[Theorem 5]{DBLP:conf/aaai/FitzsimmonsHHN19}} \\
     & $|\R|=1$: P [Theorem~\ref{thm:cand_del_oneR}] & & \\
     & OPEN (in general) & & \\
    \hline
    {\bf Possible Kemeny Score} & P [Theorem~\ref{thm:possible-kemeny-score}] & OPEN & OPEN \\
    \hline
  \end{tabularx}
  \caption{Summary of computational complexity of various manipulation schemes for \kemenyscore, \kemenyconsensus, and \kemenywinner. The computational complexity of the cases marked with ``-'' are not known to the best of our knowledge and also mentioned in~\cite{DBLP:journals/teco/KnopKM20}.}
  \label{tab:summary}
\end{table}

\paragraph{\bf Related Work.} Here, we discuss some more related work. In 1992, Bartholdi et al.~\cite{bartholdi1989voting} initiated the study of manipulating the election by adding/deleting the candidates or voters so that a favorite candidate is the winner. Afterwards, there is a lot of study on the manipulation actions for the several voting rules~\cite{brandt2016handbook,DBLP:journals/ai/BredereckFKNST21,DBLP:journals/ai/FaliszewskiMS21,DBLP:journals/toct/BredereckFNT21,DBLP:journals/jair/ChenFNT17,DBLP:conf/aldt/FitzsimmonsH15,DBLP:conf/atal/KariaMD22,DBLP:journals/tcs/AnandD21,DBLP:journals/tcs/DeyMN18,DBLP:journals/iandc/BredereckCFNN16,DBLP:conf/sagt/ElkindFS09,DBLP:journals/algorithmica/DornS12}, with several results pertaining to {\sf NP}-hardness as well as algorithms.  

Manipulation is closely linked to the robustness of voting rules, a connection identified a few years ago~\cite{DBLP:conf/sagt/BredereckFKNST17}. This insight has since inspired research into manipulation from a new perspective~\cite{DBLP:conf/aldt/GawronF19,DBLP:conf/ijcai/BoehmerBFN21,DBLP:conf/sofsem/MisraS19,faliszewski2022robustness,caragiannis2022evaluating}. This relation also highlights the importance of algorithmic results for these problems.
\paragraph{Road Map.} We begin our technical sections from Section~\ref{sec:pks}, which is dedicated to a polynomial time algorithm for \possiblewinners. Then, in Section~\ref{sec:dollar}, we present results for \dollarbribery and \voterdeletion. Section~\ref{sec:swap} and Section~\ref{sec:cd} are dedicated to \swapbribery and \candidatedeletion, respectively. Section~\ref{sec:conclusion} concludes the paper with some open questions.

\section{Preliminaries}
The notation $[n]$ denotes the set of natural numbers $\{1,\ldots,n\}$. The Kendall tau distance between two rankings $\pi,\pi'$ is denoted as $\dist{\pi}{\pi'}$. Furthermore, we abuse the notations and use $\dist{\pi}{\R}$ to denote the Kendall tau distance of ranking $\pi$ from the set of rankings $\R$. For a ranking $\pi$, let  $S(\pi)$ denote the set of elements that are ranked in $\pi$, i.e., $\pi$ is a ranking over $S(\pi)$. We also use $\dist{\pi}{\pi'}$ to denote the distance between $\pi$ and $\pi'$ when either $S(\pi)\subseteq S(\pi')$ or $S(\pi')\subseteq S(\pi)$. The distance is computed only for the elements in $S(\pi)\cap S(\pi')$.  It will be clear from the context. When we say that a ranking $\pi$ is $k$ distance away from another ranking $\pi'$ (or set of rankings $\R$), we mean that $\dist{\pi}{\pi'}\leq k$ (or $\dist{\pi}{\R}\leq k$). A pair of candidates $c,c'$ disagree in $\pi$ and $\pi'$ if their relative ordering differs in $\pi$ and $\pi'$. 
Let $\pi,\pi'$ be two rankings and $S\subseteq S(\pi)\cap S(\pi')$. We say that $\pi$ and $\pi'$ \emph{agrees} over $S$ if for every $\{c,c'\}\subseteq S$, the relative ordering of $c$ and $c'$ is same in both $\pi$ and $\pi'$, i.e., either $c$ is before $c'$ in $\pi$ and $\pi'$ or either $c'$ is before $c$ in both the rankings. Let $\pi$ be a ranking. Let $Y\subseteq S(\pi)$. Then, $\pi\vert_{Y}$ denotes the ordering $\pi$ \emph{restricted} to $Y$, i.e., it is an ordering over $Y$ and the relative ordering of every pair of elements in $Y$ is the same as that in $\pi$. We denote the $i$th element of $\pi$ as $\pi(i)$ and the number of elements in $\pi$ as $|\pi|$. The notation $x \prec_{\pi} y$ denotes that $x$ is before $y$ in the ordering $\pi$. The number of \emph{disagreements} with respect to $x$ in $\pi$ and $\pi'$ is defined as follows $\disagree{\pi}{\pi'}{x} = |\{(x,y) \colon (x \succ_\pi y \wedge y \succ_{\pi'} x) \vee (y \succ_\pi x \wedge x \succ_{\pi'} y)\}|$. We treat the set of rankings $\R$ as an ordered tuple $(R_1, \ldots, R_n)$ for convenience in the technical sections. 

For a function $f\colon P\rightarrow Q$ and a subset $S\subseteq P$, $f(S)=\cup_{s\in S}f(s)$. For a graph $G=(V,E)$, $V(G)$ denotes the vertex set of $G$ and $E(G)$ denotes the edge set of $G$. For a subset $S\subseteq V(G)$, $G-S$ denotes the graph on the vertex set $V(G)\setminus S$ and the edge set $\{uv\in E(G)\colon \{u,v\}\subseteq V(G)\setminus S\}$.  


\section{\possiblewinners}\label{sec:pks}
In this section, we design a polynomial-time algorithm for \possiblewinners problem. Recall that in \possiblewinners we need to complete the partial ranking set $\R$ so that the Kendall tau distance between the given ranking $\X$ and the extended ranking set of $\R$, say $\R'$, is at most $k$. Note that there is no cost associated with the extension of rankings. Thus, we first design an algorithm to extend a ranking $R$ over $C'\subseteq \C$ to $R'$  using the candidates in $C'' = \C \setminus C'$ such that $\dist{\X}{R'}$ is least among all the extensions of $R$. We call $R'$ as an \emph{optimal} extension of $R$. We will argue that we can extend all the rankings in $\R$ using this algorithm and return \yes if and only if we are given a yes-instance. 

The algorithm proceeds as follows. We have a ranking $R$ over $C'\subseteq \C$ as an input. Let $C'' = \C \setminus C'$ and $\Pi=\X \vert_{C''}$.
Let $R'=R$. 
For $i\in [|\C''|]$, we insert $\Pi(i)$ after the candidates $\Pi(1), \ldots, \Pi(i-1)$ in $R'$ at the leftmost position that minimizes the Kendall tau distance of $R'$  from $\X$ after inserting $\Pi(i)$. Algorithm~\ref{alg:possible-kemeny-score} describes the procedure formally. By \emph{inserting} a candidate $c$ in the ranking $R'$ at position $j$, we mean that $c$ is inserted between $R'\vert_{C'}(j)$ and $R'\vert_{C'}(j+1)$ and after all the candidates of $C''$ between $R'\vert_{C'}(j)$ and $R'\vert_{C'}(j+1)$ in $R'$. We denote this operation by $\insertop{R'}{c}{j}$. For example, let $C' = \{1, 2, \ldots, 100\}$, $C'' = \{a, b, \ldots, z\}$, and $R = 1 \prec 2 \prec 3 \prec a \prec b \prec c \prec 4 \prec \ldots \prec 100$. Then $\insertop{R}{d}{3}$ will result in ranking $R' = 1 \prec 2 \prec 3 \prec a \prec b \prec c \prec d \prec 4 \prec \ldots \prec 100$. Note that $\insertop{R'}{c}{0}$ means that $c$ is inserted before $R'\vert_{C'}(1)$ and $\insertop{R'}{c}{|C'|}$ means that $c$ is inserted after the last candidate in $R'\vert_{C'}$ in $R'$.

\begin{algorithm}[ht]
\caption{Possible-Kemeny-Score}
\label{alg:possible-kemeny-score}
\KwIn{Ranking $\X$, Partial ranking $R$ over $C' \subseteq \C$}
\KwOut{a ranking $R'$ over $\C$ which is an extension of $R$}
Initialize set $C'' = \C\setminus C'$, ranking $\Pi=\X\vert_{C''}$, and $R'_0=R$ \\
$l = 0$ \\
\For{$i\in [|\C''|]$}{
    $S = \{l,\ldots,|C'|\}$ \\
    $q=\min(\argmin\limits_{j\in S} \dist{\X}{\insertop{R'}{\Pi(i)}{j}})$ \label{step:min_pos} \\
    $R'_i=\insertop{R'_{i-1}}{\Pi(i)}{q}$ \label{step:insert} \\
    $l=q$
    
}
\Return $R'_{|C''|}$
\end{algorithm}

Before proving that $R'$ is an optimal extension of $R$, we first prove the following structural property of all optimal extensions, which is crucial for our algorithm.

\begin{lemma}
\label{lem:pks-relative-order}
Let $\hat{R}$ be an optimal extension of $R$. Then $\hat{R}$ and $\X$ agree over $C''$. 
\end{lemma}
\begin{proof}
Let $\hat{R}$ be an optimal extension of $R$. Suppose that $\{c,c'\}\subseteq C''$ such that $c \prec_\X c'$. Towards the contradiction, suppose that $c' \prec_{\hat{R}} c$. Let $R^\star$ be an ordering over $\C$ obtained by swapping $c'$ and $c$ in $\hat{R}$. Clearly, $\dist{\X}{\hat{R}} \leq \dist{\X}{R^\star}$ as $\hat{R}$ is an optimal extension. Note that if there is no candidate between $c$ and $c'$ in $\hat{R}$, then $\dist{\X}{R^\star} < \dist{\X}{\hat{R}}$, which is a contradiction. 

Next we consider the case when there exist candidates $c_1, c_2, \ldots, c_n$ between $c$ and $c'$ in $\hat{R}$, i.e., $c' \prec_{\hat{R}} \ldots \prec_{\hat{R}} c_i \prec_{\hat{R}} \ldots \prec_{\hat{R}} c$. Recall that $\{c,c'\}$ contributes to $\dist{X}{\hat{R}}$, but not to $\kt(X,R^\star)$. Let $c_i \in \{c_1, c_2, \ldots, c_n\}$. If $c_i$ is between $c$ and $c'$ in $\X$, i.e., $c \prec_\X c_i \prec_\X c'$, then note that $\{c,c_i\}$ and $\{c',c_i\}$ contribute to $\dist{\X}{\hat{R}}$, but not to $\dist{\X}{R^\star}$. If $c_i$ is after $c'$ in $\X$, i.e., $c \prec_\X c' \prec_\X c_i$, then $\{c,c_i\}$ contributes to $\dist{\X}{\hat{R}}$ and not to $\dist{\X}{R^\star}$, while $\{c',c_i\}$ contributes to $\dist{\X}{R^\star}$ and not to $\dist{\X}{\hat{R}}$. Similarly, if $c_i$ is before $c$ in $\X$, i.e., $c_i \prec_\X c \prec_\X c'$, then $\{c',c_i\}$ contributes to $\dist{\X}{\hat{R}}$ and not to $\dist{\X}{R^\star}$, while $\{c,c_i\}$ contributes to $\dist{\X}{R^\star}$ and not to $\dist{\X}{\hat{R}}$. Since we do not change the position of any other candidates, their relative ordering does not change in $R^\star$. Thus, $\dist{\X}{R^\star} < \dist{\X}{\hat{R}}$, which is a contradiction. 

Thus, every optimal extension of $R$ agrees with $\X$ over $C''$.
\end{proof}

Lemma~\ref{lem:pks-relative-order} provides the intuition for placing $\Pi(i)$ after $\Pi(1), \ldots, \Pi(i-1)$ in Algorithm~\ref{alg:possible-kemeny-score}.


Let $\tmpR$ be a ranking over $\C$, such that $\tmpR$ and $\X$ agree over $C''$. Let $\tmpR_i = \tmpR\vert_{C' \cup \{\Pi(1), \ldots ,\Pi(i)\}}$. Before analysing the correctness of Algorithm \ref{alg:possible-kemeny-score}, let us consider the following scenario. Suppose the candidate $\Pi(i)$ is positioned at $t_i$ in the ranking $\tmpR_i$, and there exists a position $t_f \geq t_i$ such that placing $\Pi(i)$ at position $t_f$ minimises the number of disagreements between $\Pi(i)$ and the candidates of $C'$ as measured between $\tmpR$ and $\X$ compared to any position after $t_i$. If we move $\Pi(i)$ along with all candidates $\ell > i, \ \Pi(\ell) \in C''$ positioned between $t_i$ and $t_f$ in $\tmpR$ to $t_f$, this operation will not increase Kendall tau distance between $\tmpR$ and $\X$ as proved in Lemma~\ref{lem:pks-greedy-placement}. Note that we conserve the relative order between candidates of $C''$ in the aforementioned operation. We illustrate this operation and the importance of Lemma~\ref{lem:pks-greedy-placement} using Figure~\ref{fig:pks}. Note that in the figure, it is clear that moving $\Pi(i)$ from $t_i$ to $t_f$ reduces the number of disagreements involving $\Pi(i)$. However, it is not obvious why moving other candidates of $C''$ in red squares does not increase the total Kendall tau distance. 

\begin{figure}
    \centering
\includegraphics[width=0.5\linewidth]{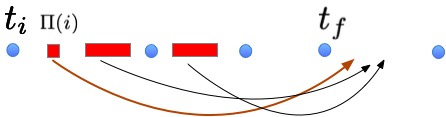}
    \caption{$\Pi(i),t_i$, and $t_f$ are as discussed in above paragraph. Circles represent the candidates of $C'$ and squares denote the candidates of $C''$. Big squares denote a chunk of $C''$ that are consecutive in the ranking.}
    \Description{Schematic figure: circles for $\mathcal{C}''$, squares for $\mathcal{C}'$,
  with one consecutive block of $\mathcal{C}''$ highlighted.}
    \label{fig:pks}
\end{figure}

\begin{lemma}
\label{lem:pks-greedy-placement}
Moving $\Pi(i)$, along with every candidate $\Pi(\ell) \in C''$ for $\ell > i$ that is positioned between $t_i$ and $t_f$ in $\tmpR$, to position $t_f$ does not increase the Kendall tau distance between $\tmpR$ and $\X$.
\end{lemma}
\begin{proof}
To streamline our arguments, we begin by defining a sequential operation that moves $\Pi(i)$ and all the candidates of $C''$ after $\Pi(i)$ that are positioned between $t_i$ and $t_f$ to $t_f$ in $\tmpR$. The procedure is as follows, first we move all the candidates of $C''$ after $\Pi(i)$ positioned at $t_i$ to $t_i + 1$, then we move all the candidates of $C''$ positioned at $t_i + 1$ to $t_i + 2$, note that in the second operation, the candidates of $C''$ positioned at $t_i + 1$ include all the candidates moved to $t_i+1$ after the first operation. We continue to do this operation until we move all the candidates $\Pi(\ell) \in C''$, $\ell \geq i$,   positioned between $t_i$ and $t_f$ to $t_f$. We denote the change in the number of disagreements of $c \in C''$ with the candidates in $C'$ when we move a candidate $c$ from its position $x$ to $y$ as $\Delta_{x, y} c$, i.e., $\Delta_{x, y} c$ is the number of disagreements when $c$ is at position $y$ minus the number of disagreements when $c$ is at position $x$. 

\begin{clm}\label{clm:best_pos}
    For every $j\in \{i,\ldots,{f-1}\}$, $\Delta_{t_j,t_j+1} \Pi(\ell)\leq \Delta_{t_j,t_{j+1}}\Pi(i)$, where $\Pi(i)$ and $\Pi(\ell)$ are at position $j$, $\ell>i$. 
\end{clm}

\begin{proof}
    Note that moving $\Pi(i)$ from $t_j$ to $t_{j+1}$, either increases or decreases the number of disagreements of $\Pi(i)$ with the candidates in $C'$ by 1.  
    We consider these two cases separately. 
\begin{description}
    \item[Case 1.] If moving $\Pi(i)$ to $t_j + 1$ increases the number of disagreements by 1 (i.e., $\Delta_{t_j, t_j + 1} \Pi(i) = +1$), this implies that, $\Pi(i) \prec_\X \tmpR\vert_{C'}(t_j + 1)$. In this case, we cannot definitively predict the effect of moving the candidate $\Pi(\ell)$ forward to $t_j + 1$, as $\Pi(i) \prec_\X \Pi(\ell)$. Therefore, $\Delta_{t_j, t_j + 1} \Pi(\ell)$ can be either positive or negative 1 when moved. Hence $\Delta_{t_j, t_j+1} \Pi(\ell) \leq \Delta_{t_j, t_j+1} \Pi(i)$

    \item[Case 2.] If moving $\Pi(i)$ to $t_j + 1$ decreases the number of disagreements by 1 (i.e., $\Delta_{t_j, t_j + 1} \Pi(i) = -1$), this implies that, $\tmpR\vert_{C'}(t_j + 1)\prec_\X \Pi(i)$. Since $\tmpR\vert_{C'}(t_j + 1) \prec_\X \Pi(i) \prec_\X \Pi(\ell)$, moving $\Pi(\ell)$ up by one position in $\tmpR$ will also decrease the number of disagreements by 1 (i.e., $\Delta_{t_j, t_j + 1} \Pi(\ell) = -1$). Hence $\Delta_{t_j, t_j+1} \Pi(\ell) \leq \Delta_{t_j, t_j+1} \Pi(i)$.
\end{description}
\end{proof}

Due to our sequential procedure and Claim~\ref{clm:best_pos}, we know that for every $\Pi(\ell)$, where $\ell>i$, at $t_{i+p}$, where $p\in \{0,\ldots,f-1\}$

\[
\sum_{0\leq j < f - p} \Delta_{t_{i + p + j}, t_{i + p + j + 1}} \Pi(\ell) \leq \sum_{0\leq j < f - p} \Delta_{t_{i + p + j}, t_{i + p + j + 1}} \Pi(i)
\]

Therefore, $\Delta_{t_{i+p},t_f}\Pi(\ell)\leq \Delta_{t_{i+p},t_f}\Pi(i)$. Due to definition of $t_f$, $\Delta_{t_{i+p},t_f}\Pi(\ell)\leq\Delta_{t_{i+p},t_f}\Pi(i)\leq 0$. Thus moving $\Pi(i)$, along with every candidate $\Pi(\ell) \in C''$ for $\ell > i$ that is positioned between $t_i$ and $t_f$ in $\tmpR$, to position $t_f$ does not increase the Kendall tau distance between $\tmpR$ and $\X$.
\end{proof}

Finally, we are ready to argue that the ranking returned by Algorithm \ref{alg:possible-kemeny-score} is an optimal extension of $R$. 

\begin{lemma}
\label{lem:pks-algorithm}
Let $R'$ be the ranking returned by Algorithm \ref{alg:possible-kemeny-score}. Then, $R'$ is an optimal extension of $R$. 
\end{lemma}
\begin{proof}
We first claim that $R'$ is an extension of $R$.
Note that $R'\vert_{C'}=R\vert_{C'}$ as we do not change the relative ordering of any pair of candidates of $C'$ in the algorithm, i.e., for all $\{c, c'\} \subseteq C'$, if $c \prec_{R} c'$, then $c \prec_{R'} c'$. Let $\Pi_i = \bigcup_{q=1}^{i} \Pi(q)$ and $R'_i$ be the ranking after $i^{th}$ iteration of Algorithm \ref{alg:possible-kemeny-score}. Next, we argue the optimality in the following claim.
\begin{clm}
\label{clm:penultimate}
For all $i \in [|C''|]$, there exists an optimal extension $R^\star$ of $R$ such that and $R^\star\vert_{C'\cup \Pi_i} = R'_i$
\end{clm}
\begin{proof}
We prove this claim using induction on $i$. For $i = 1$, consider the case where $R^\star\vert_{C'\cup \Pi_1} \neq R'_1$, which means that $t_1 \neq t'_1$ where $t_1$ is the position of $\Pi(1)$ in $R^\star$ and $t'_1$ is the position where $\Pi(1)$ was inserted in $R'_1$. If we move $\Pi(1)$ to the position $t'_1$ in $R^\star$, we note that the number of disagreements between $R^\star$ and $\X$ do not increase as $t'_1$ is the position returned by Algorithm \ref{alg:possible-kemeny-score} such that insertion of $\Pi(1)$ at that position leads to the least number of disagreements between our ranking and $\X$. Let this new ranking be $\tilde{R}$. We can see that $\tilde{R}$ is an optimal extension of $R$ such that $\tilde{R}\vert_{C'\cup \Pi_1} = R'_1$. 


Suppose that Claim \ref{clm:penultimate} holds for $i \leq \ell$. Then there exists an optimal extension $R^\star$ of $R$ such that $R^\star\vert_{C'\cup \Pi_l} = R'_l$. Next we argue for $i = \ell + 1$. Consider the case where $R^\star\vert_{C'\cup \Pi_{\ell+1}} \neq R'_{\ell+1}$, which means that $t_{\ell+1} \neq t'_{\ell+1}$ where $t_{\ell+1}$ is the position of $\Pi(\ell+1)$ in $R^\star$ and $t'_{\ell+1}$ is the position where $\Pi(\ell+1)$ was inserted in $R'_{\ell+1}$. So let us assume the case where $t'_{\ell+1} < t_{\ell+1}$. Note that $t_\ell = t'_\ell \leq t'_{\ell+1}$ as $R^\star\vert_{C'\cup \Pi_l} = R'_\ell$. We can always move $\Pi(\ell+1)$ to $t_{\ell+1}'$ in $R^\star$ and it does not change the relative ordering of candidates of $C''$. 
If we move $\Pi(\ell+1)$ to the position $t'_{\ell+1}$ in $R^\star$, we note that the number of disagreements between $R^\star$ and $\X$ involving $R^\star$ does not increase as $t'_{\ell+1}$ is the position returned by Algorithm \ref{alg:possible-kemeny-score} such that insertion of $\Pi(\ell + 1)$ at that position leads to the least number of disagreements between our ranking and $\X$, note that we are only moving $\Pi(\ell+1)$ in $R^\star$, so the changes in the Kendall tau distance will only be caused due to $\Pi(\ell+1)$. Hence, the Kendall tau distance does not increase. Let this new ranking be $\tilde{R}$, we can see that $\tilde{R}$ is an optimal extension of $R$ such that $\tilde{R}\vert_{C'\cup \Pi_{\ell+1}} = R'_{\ell+1}$. 

For the second case where $t_{\ell+1} < t'_{\ell+1}$, we can say that moving $\Pi(\ell+1)$ and every candidate $\Pi(x), \ x > \ell+1$, positioned between $t_{\ell+1}$ and $t'_{\ell+1}$ in $R^\star$ to $t'_{\ell+1}$ in $R^\star$ does not increase the number of disagreements between $R^\star$ and $\X$ as a direct consequence of Lemma \ref{lem:pks-greedy-placement}. Note that $t'_{\ell+1}$ is position returned by Algorithm \ref{alg:possible-kemeny-score} such that insertion of $\Pi(\ell + 1)$ at that position leads to the least number of disagreements between our ranking and $\X$. Let this new ranking be $\tilde{R}$, we can see that $\tilde{R}$ is an optimal extension of $R$ such that $\tilde{R}\vert_{C'\cup \Pi_{\ell+1}} = R'_{\ell+1}$. 
%
\end{proof}

As Claim \ref{clm:penultimate} holds for all $i \in [|C''|]$, we can say that there exists an optimal extension $R^\star$ of $R$ such that  $R^\star\vert_{C'\cup C'' = \C} = R'$. This proves that $R'$ returned by Algorithm \ref{alg:possible-kemeny-score} is an optimal extension to $R$. 
\end{proof}

Lemma~\ref{lem:pks-algorithm} gives an optimal extension of a ranking in $\R$. We complete the algorithm in the following theorem by extending all the rankings using the same procedure and argue the running time. 

\begin{theorem}\label{thm:possible-kemeny-score}
\possiblewinners can be solved in polynomial time.
\end{theorem}

\begin{proof}
Let $\instanceI=(\C,\R,\X,k)$ be an instance of \possiblewinners. Recall that rankings in $\R$ are over $C' \subseteq \C$. Let $\R = (R_1, R_2, \ldots, R_n)$. For every $R_i\in \R$, we find an optimal extension $R'_i$ using Algorithm~\ref{alg:possible-kemeny-score}. Let $\R'=(R'_1, R'_2, \ldots, R'_n)$ be the set of optimal extensions obtained by applying Algorithm \ref{alg:possible-kemeny-score} to each ranking in $\R$. If $\dist{\X}{\R} \leq k$, we return \yes; otherwise \no. Next, we prove the correctness of our algorithm. Clearly, if we return \yes, then $\instanceI$ is a yes-instance of \possiblewinners. We next argue that if $\instanceI$ is a yes-instance of \possiblewinners we return an extension $\R'$ such that $\dist{\X}{\R'} \leq k$. Let us suppose that there exists a set of extended rankings $\R^\star = (R^\star_1, R^\star_2, \ldots ,R^\star_n)$, such that $\dist{\X}{\R^\star} \leq k$. Due to Lemma \ref{lem:pks-algorithm}, we can say that for all $i \in [n]$, $\dist{\X}{\R'_i} \leq \dist{\X}{R^\star_i}$, which implies that, $\dist{\X}{\R'} = \sum_i \dist{\X}{\R'_i} \leq \sum_i \dist{\X}{R^\star_i} = \dist{\X}{\R^\star} \leq k$.

Next, we argue the running time of the algorithm. Note that Algorithm~\ref{alg:possible-kemeny-score} runs in time $\mathcal{O}(m^3)$ and we call Algorithm~\ref{alg:possible-kemeny-score} $n$ times. Thus, the algorithm runs in time $\mathcal{O}(n \cdot m^3)$. Note that in Algorithm~\ref{alg:possible-kemeny-score} at line number $5$, the Kendall tau distance calculation runs in time $\mathcal{O}(m)$ because we only need to compute the number of disagreements between $\Pi(i)$ and the candidates of $C'$.

This completes the proof.
\end{proof}

    In our algorithm, we assume that all the rankings in $\R$ are over the subset $C'\subseteq \C$. However, this is only for ease of explanation, and the algorithm works even when rankings in $\R$ are not over the same subset. This is due to the reason that we call Algorithm~\ref{alg:possible-kemeny-score} independently for each partial ranking in $\R$.  


\section{\$-Bribery-Kemeny Score and Ranking Deletion Kemeny Score}\label{sec:dollar}
We begin with designing a polynomial-time algorithm for \dollarbribery. Recall that in \dollarbribery we need to change some rankings to $\X$ so that $\X$ is at most $k$ distance away from the new set of rankings. The problem inherently seems to be closer to the {\sc Knapsack} problem, in which given a set of items, $A=\{a_1,\ldots,a_n\}$, a weight function $w\colon A \rightarrow \mathbb{N}$, a value function $v\colon A\rightarrow \mathbb{N}$, and two natural numbers $T$, $W$; the goal is to decide if there exists a subset $S\subseteq A$ such that $\sum_{a\in S}w(a)\leq W$ and $\sum_{a\in S}v(a)\geq T$. We formalize this intuition below by giving a polynomial time reduction to the {\sc Knapsack} problem.

\begin{lemma}
\label{lem:red_dollarbribery-to-knapsack}
There is a polynomial time reduction from the \dollarbribery problem to the {\sc Knapsack} problem, where the values are polynomially bounded in the input size.
\end{lemma}

\begin{proof}
Give an instance $\instanceI = (\C, \R, \X, \cost, \budget, k)$ of the \dollarbribery problem, we create an instance $\instanceJ = (A,w,v$ $,T,W)$ of the {\sc Knapsack} problem as follows. Let $\R=(R_1,\ldots,R_n)$. The set of items is $A=\{a_1,\ldots,a_n\}$. The weight function $w\colon A\rightarrow \mathbb{N}$ is defined as follows: $w(a_i)=\cost(R_i)$ and the value function $v\colon A\rightarrow \mathbb{N}$ is defined as follows: $v(a_i)=\dist{\X}{R_i}$. 
Furthermore,  $W=\budget$ and $T = D - k$, where $D=\dist{\X}{\R}$. Note that the value of every item is bounded by $|\C|^2$ and $T$ is at most $n|\C|^2$.  Next, we prove the correctness.

In the forward direction, suppose that $\instanceI$ is a yes-instance of \dollarbribery. Let $S = \R\setminus\R'$ where $\R'$ is the new set of rankings, i.e., $S$ is the set of rankings in $\R$ that are manipulated.  Note that $\dist{\X}{\R'}\leq k$ and  $\sum_{R \in S}{\cost(R)} \leq \budget$. We construct a set $S'$ as follows: $S'=\{a_i\in A \colon R_i \in S\}$. We claim that $S'$ is a solution to $\instanceJ$. Due to the construction, we know that $\sum_{a\in S'}w(a)\leq \budget$. As discussed in the definition of \dollarbribery, we can assume that every ranking in $S$ is the same as $\X$. Thus, $\dist{\X}{\R'}=\dist{\X}{\R\setminus S}$.  Hence, $\dist{\X}{S}+\dist{\X}{\R'}=D$. Since $\dist{\X}{\R'}\leq k$, it follows that $\dist{\X}{S}\geq D-k$. Thus, $\sum_{a\in S'}v(a)\geq D-k$. 

In the reverse direction, suppose that $\instanceJ$ is a yes-instance of the {\sc Knapsack} problem. Let $S\subseteq A$ be a solution to $\instanceJ$. We construct a set of rankings $S'$ as follows: $S'=\{R_i\in \R \colon a_i\in S\}$. We claim that $S'$ is a solution to $\instanceI$. Let $\R'$ be the rankings obtained by changing $S'$ to $\X$ and the remaining rankings are same in $\R$ and $\R'$.  Clearly, $\sum_{R\in S'}\cost(R)\leq \budget$. Since, $\sum_{a\in S}v(a)\geq D-k$, using the same arguments as earlier, $\dist{\X}{\R'}\leq k$. 
This completes the proof.
\end{proof}

Since there is a pseudo-polynomial time algorithm for {\sc Knapsack} that solves the problem in polynomial time if either weights or values are  polynomially bounded~\cite{kleinberg2006algorithm} and the above reduction is polynomial time, we have the following result.  

\begin{theorem}\label{thm:dollar-bribery}
\dollarbribery can be solved in polynomial time. %
\end{theorem}

Note that changing a ranking to $\X$ is equivalent to deleting this ranking as the contribution of this ranking to Kendall tau distance becomes $0$ due to both the operations. Thus, we have the following result due to the same algorithm, with the difference that instead of changing it to $\X$, we delete the ranking.

\begin{theorem}\label{thm:ranking-deletion}
\voterdeletion can be solved in polynomial time. %
\end{theorem}

\section{Swap Bribery-Kemeny Score}\label{sec:swap}
In this section, we design a polynomial-time algorithm for \swapbribery problem. Recall that in \swapbribery we need to perform some admissible swaps on the set of rankings $\R$, so that $\X$ is at most $k$ distance away from the new set of rankings $\R'$. Let $\tmpR$ be a ranking over $\C$. Consider a pair of candidates $c,c'$ such that  $c\prec_X c'$ and $c'\prec_\tmpR c$. If $c,c'$ are consecutive in $\tmpR$, then we say that $\{c,c'\}$ is an admissible disagreement between $\tmpR$ and $\X$, otherwise, it is a non-admissible disagreement. In our algorithm, we swap candidates only when they are in disagreement. Before presenting the algorithm, we first establish an important property regarding admissible disagreements. Specifically, if there exists a non-admissible disagreement $\{a, b\}$ between rankings $\tmpR$ and $\X$, then it implies the existence of an admissible disagreement between $\tmpR$ and $\X$. This property is crucial because it implies that for any ranking $\tmpR$, either there exists an admissible disagreement or there is no disagreement between $\tmpR$ and $\X$ (i.e., $\dist{\X}{\tmpR} = 0$). Consequently, this ensures that in our algorithm, we can always perform an admissible swap action if $\dist{\X}{\tmpR} > 0$. We formalise all these properties below.

Throughout this section, when we refer to \emph{performing a swap action between candidates}, we implicitly mean that the candidates are in disagreement. Throughout the section, after performing a swap action on a ranking $R$, we reuse the notation $R$ for the new ranking. 
We begin with the following observation.

\begin{obs}
\label{obs:swap}
One swap operation in $\tmpR$ of an admissible disagreement between $\tmpR$ and $\X$ reduces the Kendall tau distance between them by exactly one.
\end{obs}

\begin{lemma}
\label{lem:admissible_disagreement_existance}
The presence of a non-admissible disagreement between rankings $\tmpR$ and $\X$ implies the existence of an admissible disagreement between $\tmpR$ and $\X$. 
\end{lemma}
\begin{proof}
Suppose that there exists a non-admissible disagreement $\{a, b\}$ between $\tmpR$ and $\X$. Towards the contradiction, suppose that there are no admissible disagreements between $\tmpR$ and $\X$. Since $\{a, b\}$ is a disagreement, let $a \prec_\tmpR b$ and  $b \prec_\X a$. Note that $\{a, b\}$ is a non-admissible  disagreement, so there must be some candidates $c_1, c_2, \ldots, c_p$ ranked between $a$ and $b$ in $\tmpR$. Let $\tmpR = \ldots a \prec c_1 \prec \ldots \prec c_p \prec b \prec \ldots$. Since there is no admissible disagreement between $\tmpR$ and $\X$, we know that $a \prec_\X c_1$, $ c_i \prec_\X c_{i+1}$, for all $i\in [p-1]$, and $c_p \prec_\X b$. This implies $a \prec_\X b$ due to transitivity in the ranking $\X$, which is a contradiction.
\end{proof}

Due to Lemma~\ref{lem:admissible_disagreement_existance}, we have the following useful property. 

\begin{corollary}
\label{obs:algorithm}
We can always perform an admissible swap action in the ranking $\tmpR$ if $\dist{\X}{\tmpR} > 0$.
\end{corollary}
By performing $s$ swaps in the ranking $\tmpR$, we mean executing $s$ swap operations sequentially, each applied to the updated ranking resulting from the previous swap. It is important to note that a swap which is initially non-admissible may become admissible after performing some admissible swaps, allowing us to perform that swap at a later stage. Due to Observation~\ref{obs:swap} and Corollary~\ref{obs:algorithm}, we can always perform a sequence of $0 \leq s \leq \dist{\X}{\tmpR}$ admissible swap actions in $\tmpR$ and that results in a drop of $s$ in the Kendall tau distance between $\tmpR$ and $\X$.



We finally move towards designing our algorithm.

\begin{theorem}\label{thm:swap-bribery}
\swapbribery can be solved in polynomial time.
\end{theorem}
\begin{proof}
Let $\instanceI = (\C,\R,\X,\cost,\budget,k)$ be a given instance of the \swapbribery problem. We design a dynamic programming algorithm to solve this problem. Let $\R = (R_1, R_2, \ldots, R_n)$. We define the dynamic programming table as follows: $T[i, s_1, s_2, j] = $ minimum cost incurred in performing admissible swap actions on the first $i$ rankings such that $s_1$ swaps are done to $R_i$, $s_2$ swaps are done to the set of rankings ($R_1, R_2, \ldots, R_{i-1}$), and the Kendall tau distance between the set of rankings ($R_1, R_2, \ldots, R_{i}$) and $\X$ is at most $j$. 

For $i=1$, there are no previous rankings to consider; however, we write $T[1,s_1,s_2,j]$ for brevity, and for $s_2>0$, it is an invalid entry. For $1 < i \leq |\R|$, if either $s_1 > \dist{\X}{R_i}$ or $s_2 > \dist{\X}{(R_1, R_2, \ldots, R_{i-1})}$, then also  the entry $T[1,s_1,s_2,j]$  is invalid. 
 We store $\infty$ for the invalid entries or if the allowed number of swaps is insufficient to achieve the desired Kendall tau distance. 

We compute the remaining table entries as follows.

\textbf{Base Case:} 
For $i = 1$, $0 \leq s_1 \leq \dist{\X}{R_1}$, $s_2 = 0$ and $0 \leq j \leq k$, we compute the table as follows.
\begin{equation}
\label{dp:swap-bribery-base}
T[1, s_1, 0, j] = 
\begin{cases} 
s_1 \cdot \cost(R_1) & \text{If} \  \dist{\X}{R_1} - s_1 \leq j\\
\infty & \text{otherwise}.
\end{cases}
\end{equation}


\textbf{Recursive Step:} For all $1 < i \leq |\R|$, $0 \leq s_1 \leq \dist{\X}{R_i}$, $0 \leq s_2 \leq \dist{\X}{(R_1, R_2, \ldots, R_{i-1})}$, and $0 \leq j \leq k$, we compute the table entries as follows.
\begin{equation}
\label{dp:swap-bribery-recursive}
\begin{split}
T[i, s_1, s_2, j] = \min_{s_{2}' + s_{2}'' = s_2} & T[i-1, s_{2}', s_{2}'', j - \dist{\X}{R_i} + s_1] \\
& + s_1 \cdot \cost(R_i)
\end{split}
\end{equation}

Next, we prove the correctness of this dynamic programming algorithm in the following lemma.

\begin{lemma}
\label{lem:swap-bribery-dp-correctness}
For every $i \in [|\R|]$, $0 \leq s_1 \leq \dist{\X}{R_i}$, $0 \leq s_2 \leq \dist{\X}{(R_1, R_2, \ldots, R_{i-1})}$, and $0 \leq j \leq k$, Equations~\ref{dp:swap-bribery-base} and \ref{dp:swap-bribery-recursive} correctly compute $T[i, s_1, s_2, j]$.
\end{lemma}
\begin{proof}
We will prove the correctness of the dynamic programming solution through induction on $i$. For the base case, when $i = 1$, if $\dist{\X}{R_1}-s_1 \leq j$, then by performing $s_1$ swaps, the distance between the resultant ranking and $\X$ is at most $j$, otherwise $s_1$ swaps does not decrease the distance to $j$. Hence, $T[1,s_1,0,j]$ is computed correctly by Equation~\ref{dp:swap-bribery-base}. 

For the inductive step, let us suppose that the table entries are computed correctly for all $i \leq l$, $0 \leq s_1 \leq \dist{\X}{R_i}$, $0 \leq s_2 \leq \dist{\X}{(R_1, R_2, \ldots, R_i)}$ and $0 \leq j \leq k$. We denote the value of $T[l+1, s_1, s_2, j]$ computed by Equation~\ref{dp:swap-bribery-recursive} as $alg$. Let $\opt$ be a set of swap operations corresponding to the minimum cost incurred in performing admissible swap actions on the first $l+1$ rankings such that $s_1$ swaps are done to $R_{l+1}$, $s_2$ swaps are done to the set of rankings $(R_1, R_2, \ldots, R_l)$, and the Kendall tau distance between the set of rankings $(R_1, R_2, \ldots, R_{l+1})$ and $\X$ is at most $j$. Let $\opt_{\cost}$ be the cost of performing swaps in $\opt$. Clearly, $\opt_{\cost} \leq alg$. Let us suppose that $s_l$ admissible swaps are performed on $R_l$ in $\opt$, this implies that $s_2 - s_l$ admissible swaps are performed on $(R_1, R_2, \ldots, R_{l-1})$. We claim that $T[l, s_l, s_2 - s_l, j - \dist{\X}{R_{l+1}} + s_1] + s_1 \cdot \cost(R_{l+1}) = \opt_{\cost}$. 


Due to Equation~\ref{dp:swap-bribery-recursive},  $alg \leq T[l, s_l, s_2 - s_l, j - \dist{\X}{R_{l+1}} + s_1] + s_1 \cdot \cost(R_{l+1})$. Since $alg\geq \opt_{\cost}$,  $T[l, s_l, s_2 - s_l, j - \dist{\X}{R_{l+1}} + s_1] + s_1 \cdot \cost(R_{l+1}) \geq \opt_{\cost}$. 
%
Since $\opt_{\cost} - s_1 \cdot \cost(R_{l+1})$ is the cost of performing $s_2$ admissible swaps on the set of rankings $(R_1, R_2, \ldots, R_l)$ where $s_l$ admissible swaps are performed on $R_l$ and $s_2 - s_l$ swaps are performed on the set of rankings $(R_1, R_2, \ldots, R_{l-1})$ such that the Kendall tau distance is at most $j - (\dist{\X}{R_{l+1}} - s_1) = j - \dist{\X}{R_{l+1}} + s_1$ from $\X$, due to induction hypothesis, $T[l, s_l, s_2 - s_l, j - \dist{\X}{R_{l+1}} + s_1]  \leq \opt_{\cost} - s_1 \cdot \cost(R_{l+1})$. Hence, $T[l, s_l, s_2 - s_l, j - \dist{\X}{R_{l+1}} + s_1] + s_1 \cdot \cost(R_{l+1}) = \opt_{\cost}$. This implies that $alg \leq \opt_{\cost}$. Hence $alg = \opt_{\cost}$. Thus Equations~\ref{dp:swap-bribery-base} and \ref{dp:swap-bribery-recursive} correctly compute the dynamic programming table.
\end{proof}
The minimum cost incurred in performing admissible swap actions on $\R$, such that the Kendall tau distance between $\R$ and $\X$ is at most $j$ is $\min\limits_{\substack{s_1, s_2}} T[|\R|, s_1, s_2, k]$. So we return \yes if $\min\limits_{\substack{s_1, s_2}} T[|\R|, s_1, s_2, k]$ $ \leq \budget$, otherwise \no. The correctness follows due to Lemma~\ref{lem:swap-bribery-dp-correctness}
 
We compute at most $|\R|^2m^2k$ entries and each entry can be computed in time $\mathcal{O}(m^2)$. Thus, the algorithm runs in polynomial time.
\end{proof}


\section{Candidate Deletion Kemeny Score}\label{sec:cd}

In this section, we begin with the case when $k=0$, and design a polynomial time algorithm by reducing the problem to the {\sc Heaviest Common Increasing Subsequence} (\mwcis) on permutations problem. In the \mwcis problem, given $p$ sequences of numbers, a weight function $w$ on the numbers, and an integer $W$, the goal is to decide whether there exists a sequence of numbers with a total weight of at least $W$ that is an increasing subsequence of each of the given $p$ sequences. In our reduction, these sequences are permutations. 

\begin{lemma}
There exists a polynomial time reduction from \candidatedeletion to \mwcis on permutations when $k=0$.
\end{lemma}

\begin{proof}
Let $\instanceI=(\C,\R,\X,\cost,\budget,k)$ be an instance of \candidatedeletion where $k=0$. Let $\R=(R_1,\ldots,R_n)$. We define a function $f\colon \C \rightarrow [m]$ as follows: $f(\X(i))=i$. Recall that $\X$ is a permutation over $\C$ and $\C$ has $m$ candidates. We construct a set of permutations $J_1,\ldots,J_n$ as follows. Each $J_i$ is an $m$-length sequence, where for $j\in [m]$, $J_i(j)=f(R_i(j))$. We can observe that $J_i(j) \neq J_i(k)$ when $j \neq k$ as $R_i(j) \neq R_i(k)$ when $j \neq k$ and $f$ is injective, which implies that the elements of $J_i$ are unique for all $i$ and due to the construction the maximum value of $f$ and the length of $J_i$ are both $m$. Hence $J_i$ is a permutation for all $i$. The weight of number $i$ is $w(i) = \cost(f^{-1}(i))$. Let $W=\sum_{c\in \C}\cost(c)-\budget$. Let $\instanceJ=(J_1,\ldots,J_n, w, W)$ be an instance of \mwcis on permutations. Next, we prove the equivalence between the two instances. 

In the forward direction, let $Z$ be a solution to $\instanceI$. Let $\tilde{Z}=f(\C \setminus Z)$ We claim that $J_1\lvert_{\tilde{Z}}$ is a solution to $\instanceJ$. Since $\sum_{c\in Z}\cost(c)\leq \budget$, $\sum_{c\in \C \setminus Z}\cost(c) \geq W$. Thus $\sum_{i \in \tilde{Z}} w(i) \geq W$. Suppose that $J_1\lvert_{\tilde{Z}}$ is not an increasing subsequence of the permutation $J_i$, $i\in [n]$. Then,  there exists $\{\ell , \ell'\}\subseteq \tilde{Z}$ such that $\ell < \ell'$, and $\ell'$ is before $\ell$ in $J_i$. Since $\ell < \ell'$, it follows that $f^{-1}(\ell)\prec_X f^{-1}(\ell')$. Since $\ell'$ is before $\ell$ in $J_i$ it follows that $f^{-1}(\ell') \prec_{R_i} f^{-1}(\ell')$.  Furthermore, since $\{f^{-1}(\ell),f^{-1}(\ell')\}\subseteq \C\setminus Z$, it contradicts that $Z$ is a solution to $\instanceI$. 


In the reverse direction, let $\Pi$ be a solution to $\instanceJ$. Let $S$  be the set of numbers that are in $\Pi$ and $\tilde{S}$ be the set of numbers that are not in $\Pi$. Let $Z=\{f^{-1}(s)\colon s\in \tilde{S}\}$. We claim that $Z$ is a solution to $\instanceI$. Since $\sum_{i \in S} w(i) \geq W$ it follows that $\sum_{c\in Z}\cost(c)\leq \budget$. Suppose $Z$ is not a solution to $\instanceI$, then there exists a pair of candidates $\{c,c'\} \subseteq \C \setminus Z$ and a ranking $R_i$ such that $c\prec_X c'$ and $c' \prec_{R_i} c$. Suppose that $c=X(j)$ and $c'=X(j')$. Thus $j<j'$ and $j'$ is before $j$ in the sequence $J_i$. Since $\{c,c'\} \subseteq \C\setminus Z$ it follows that$\{j,j'\} \subseteq S$.
This contradicts the fact that $\Pi$ is an increasing common subsequence of $J_1,\ldots,J_n$. 
\end{proof}

Chan et al.~\cite{DBLP:conf/isaac/ChanZFYZ05} designed an algorithm for \mwcis with unit weights. 
The same algorithm can be used for non-unit weights by storing the maximum weight of the common increasing subsequence as the rank of a match instead of the longest length. We present the complete algorithm for completeness. 

\begin{lemma}
     \mwcis can be solved in polynomial time when the sequences are permutations.
\end{lemma}
\begin{proof}

Let $\instanceI = (J_1, J_2, \ldots,  J_m, w, W)$ be the given instance of the problem, where for all $i \in [m]$, $J_i$ is a permutation over $[n]$ and for all $x \in [n]$, $w(x) \geq 0$. Let $J_i^{-1}(x)$ denote the position of $x$ in $J_i$ and let $J^k_i$ denote the prefix of $J_i$ with k elements.

We define the notations \emph{match}, \emph{dominating match}, and \emph{rank} as follows. A match is a $m$-tuple, $(i_1, \ldots, i_m)$ where $J_1(i_1) = J_2(i_2) = \ldots = J_m(i_m)$, the value of this match is $J_1(i_1)$. A match $(\delta_1, \ldots, \delta_m)$ is called a dominating match of another match $(\theta_1, \ldots, \theta_m)$ if $J_1(\delta_1)$ $< J_1(\theta_1)$ and $\delta_j < \theta_j$ for all $1 \leq j \leq m$. Note that $\Delta_j$ denotes the match with the value of $j$ and $\Delta_j$ is a unique match because our sequences are permutations. The rank of a match $\Delta = (\delta_1, \ldots, \delta_m)$, denoted by $R(\Delta)$, is the weight of the \mwcis of $J^{\delta_1}_1, J^{\delta_2}_2, \ldots, J^{\delta_m}_m$ such that $J_1(\delta_1)$ is the last element of the \mwcis. For $\Delta_1, \ldots, \Delta_n$, $R(\Delta_j)$ is computed in increasing order of $j$ as follows,

\begin{equation}
\label{dp:hcis}
R(\Delta_j) = 
\begin{cases} 
w(\Delta_j) & \text{if no match} \\ & \text{dominates } \Delta_j.\\
w(\Delta_j) + max\{R(\Delta_k \vert \Delta_k \text{ dominates } \Delta_j\} & \text{otherwise}.\\
\end{cases}
\end{equation}
It is easy to see that \mwcis would be $\max_{\Delta_j} R(\Delta_j)$. For the decision version, we accept iff $\max_{\Delta_j} R(\Delta_j) \geq W$.

We will now prove the correctness of the dynamic programming solution through induction on $j$. Note that because the sequences are permutations, there are $n$ matches with values $1$, $2$, $\ldots$, $n$. We can compute all the matches in time $\mathcal{O}(m \cdot n)$ and we process them in the increasing order of their value i.e., for a match with value $j$, the ranks of all the matches with value $1, \ldots, j-1$ should have been computed first as a match with value $j$ can only be dominated by a match which has a value less than $j$.

For the base case, $R(\Delta_1) = w(\Delta_1) = w(1)$ as there is no match that can dominate $\Delta_1$ and hence the \mwcis ending with value $1$ only has one element with weight $w(1)$. Let us suppose that for all $y < j$, $R(\Delta_y)$ is computed correctly. For the first case assume $w_{OPT}$ as the optimal rank of match $\Delta_j$ that is the weight of the sequence $OPT = j$ i.e. the \mwcis ending at match $\Delta_j$ contains only the element $j$. It is obvious in this case that the rank computed by Equation \ref{dp:hcis} $= w_{OPT}$. For the second case assume $w_{OPT}$ as the optimal rank of the match $\Delta_j$ that is weight of the sequence $OPT = \ldots, x, j$. We know that the rank of match $\Delta_j$ computed by Equation \ref{dp:hcis} $\leq w_{OPT}$ as $w_{OPT}$ is the optimal rank of the match $\Delta_j$. Since $x$ is before $j$ in $OPT$, $x$ has to be before $j$ in every sequence $J_i$ and $x < j$, hence $\Delta_x$ dominates $\Delta_j$ which means that the rank of match $\Delta_j$ computed by Equation \ref{dp:hcis} $\geq w(\Delta_j) + R(\Delta_x) \geq w_{OPT}$. Hence the rank of match $\Delta_j$ computed by Equation \ref{dp:hcis} $= w_{OPT}$. Hence Equation \ref{dp:hcis} correctly computes the rank of all matches.

We precompute all the $n$ matches in time $\mathcal{O}(m \cdot n)$, then for each match we check all the $n$ matches for domination in time $\mathcal{O}(m)$, hence the algorithm runs in time  $\mathcal{O}(m \cdot n^2)$



\end{proof}

\balance

\begin{theorem}\label{thm:cand_del_polyk0}
    \candidatedeletion can be solved in polynomial time when $k=0$. 
\end{theorem}

Next, we show that for one ranking (i.e., $|\R|=1$), the problem is equivalent to the  
{\sc Partial Vertex Cover} ({\sc WPVC}) problem for permutation graphs. In {\sc WPVC}, given a vertex-weighted graph $G$ with weight function $w\colon V(G)\rightarrow \mathbb{N}$, and integers $W,t$, the goal is to find a subset $S\subseteq V(G)$ of weight at most  $W$ that covers at least $t$ edges, i.e., for at least $t$ edges, at least one of the endpoints is in $S$. 
A \emph{permutation graph} is a graph whose vertices represent the elements of a permutation, and whose edges represent pairs of elements that are reversed by the permutation, and we say that such a permutation realises the given permutation graph. The complexity of {\sc WPVC} for permutation graph is not known. However, if {\sc WPVC} for permutation graph is \nph, then \candidatedeletion is also \nph; and if it is polynomial-time solvable, then \candidatedeletion is polynomial time solvable for one ranking. 

\begin{lemma}\label{lem:wpvc_cd}
    There is a polynomial time reduction from {\sc WPVC} on permutation graph to \candidatedeletion. 
\end{lemma}

\begin{proof}
    Let $\instanceI=(G,w,W,t)$ be an instance of {\sc WPVC} where $G$ is a permutation graph. Let $\Pi$ be a permutation realising the permutation graph $G$. Let $n=|V(G)|$. Let $\C=V(G)$, $X=1 \succ 2 \succ \ldots \succ n$, $R=\Pi$, $\cost=w$, $\budget=W$,  and $k=|E(G)|-t$. Let $\instanceJ=(\C,\R,X,\cost,\budget,k)$ be an instance of \candidatedeletion. Next, we prove the equivalence between the two instances. In the forward direction, let $Z$ be a solution to $\instanceI$. We claim that $Z$ is also a solution to $\instanceJ$. Note that the number of edges in $E(G-Z)$ is at most $|E(G)|-t$, thus due to the definition of permutation graph, the number of pairs of vertices such that $i<j$ and $\Pi(i)>\Pi(j)$, where $\{\Pi(i),\Pi(j)\} \subseteq G-Z$, is at most $|E(G)|-t$. Thus, $\kt(X\lvert_{\C\setminus Z},R\lvert_{\C\setminus Z})\leq |E(G)|-t$. Clearly, $\cost(Z) \leq \budget$. In the reverse direction, let $Z$ be a solution to $\instanceJ$. Then, we claim that $Z$ is also a solution to $\instanceI$. Clearly, $w(Z)\leq W$. Note that $\kt(X,R)=|E(G)|$ by the definition of permutation graph. Since $\kt(X\lvert_{\C\setminus Z},R\lvert_{\C\setminus Z})\leq |E(G)|-t$, for at least $t$ disagreements, one of the candidate is in $Z$. Thus, $Z$ covers at least $t$ edges in $G$.   
\end{proof}

Due to Lemma~\ref{lem:wpvc_cd}, if {\sc WPVC} on a permutation graph is \nph, then \candidatedeletion is also \nph. Obtaining a polynomial-time algorithm for {\sc WPVC} on a permutation graph is also beneficial due to the following result. 

\begin{lemma}\label{lem:cd_wpvc}
    There is a polynomial time reduction from \candidatedeletion when $|\R|=1$ to {\sc WPVC} on permutation graph.
\end{lemma}

\begin{proof}
    Let $\instanceI=(\C,\R,\X,\cost,\budget,k)$ be an instance of \candidatedeletion where $\R$ has only one ranking, say $R$. We construct an instance $\instanceJ = (G, w, W,t)$ of the {\sc WPVC} problem as follows. Let $f\colon \C \rightarrow [m]$ be as function defined as follows: for each $i\in [m]$, $f(\X(i))=i$. Let $\Pi=f(R(1)),\ldots,f(R(m))$ be a permutation and $G_\Pi$ be a permutation graph corresponding to the permutation $\Pi$. Let $w=\cost$, $W=\budget$, and $t=|E(G)|-k$. Let $\instanceJ=(G_\Pi,w,W,t)$ be an instance of {\sc WPVC} where $G_\Pi$ is a permutation graph. Next, we show the equivalence between the two instances. 

    In the forward direction, let $Z$ be a solution to $\instanceI$. We claim that $Z$ is also a solution to $\instanceJ$. Clearly, $w(Z)\leq W$. Since $|E(G)|=\kt(X,R)$ and $\kt(X\lvert_{\C\setminus Z},R\lvert_{\C\setminus Z})\leq k$, for at least $\kt(X,R)-k$ edges one of the endpoint is in $Z$. In the reverse direction, suppose that $Z$ is a solution to $\instanceJ$. Then, $G-Z$ has at most $|E(G)|-k$ edges. Note that if there exists a pair of candidates $x,y$ that disagree between $R$ and $\X$, then $f(x)$ and $f(y)$ are reversed and there is an edge between then in $G_\Pi$. Thus, there can be at most $k=|E(G)|-t$ disagreements between $R$ and $\X$. This completes the proof.  
\end{proof}

Lemma \ref{lem:wpvc_cd} gives a polynomial‐time reduction from WPVC on permutation graphs to \candidatedeletion, and the fact that our construction fixes $\vert R\vert=1$ does not weaken the result, any instance of WPVC is still mapped to an equivalent instance of the \candidatedeletion. Conversely, Lemma \ref{lem:cd_wpvc} reduces that same $\vert R\vert=1$ variant back to WPVC.
Because both reductions operate on the identical restricted case, they establish a full polynomial-time equivalence, which not only yields Theorem \ref {thm:cand_del_oneR} but also transfers hardness: if WPVC on permutation graphs were shown NP-hard, the same would hold for \candidatedeletion with $\vert R \vert=1$ and hence for \candidatedeletion in general

Thus, due to Lemma~\ref{lem:wpvc_cd}, if {\sc WPVC} on a permutation graph is solvable in polynomial time, then \candidatedeletion is also solvable in polynomial time for one ranking. Thus, due to Lemma~\ref{lem:wpvc_cd} and~\ref{lem:cd_wpvc}, we have the following result. 

\begin{theorem}\label{thm:cand_del_oneR}
    \candidatedeletion can be solved in polynomial time for one ranking if and only if {\sc WPVC} can be solved in polynomial time for permutation graphs unless ${\sf P}={\sf NP}$. 
\end{theorem}

\section{Conclusion}\label{sec:conclusion}

In this paper, we studied the \kemenyscore problem under some well-studied bribery schemes. Unlike many other voting rules, we design polynomial time algorithms for \$-bribery, swap bribery, voter deletion, special case of candidate deletion when we want Kendall tau distance as 0, and possible Kemeny score where we extend the given partial rankings. Our main open question is to resolve the computational complexity of \candidatedeletion. We believe that it is \nph. Another research direction is to study the bribery actions for Kemeny winner. It has been studied for candidate/voter addition/deletion in the literature. It would be interesting to resolve the complexity of bribery action.

\bibliographystyle{ACM-Reference-Format} 
\bibliography{sample}


\end{document}